\documentclass[proceedings]{aofa}
\usepackage[utf8]{inputenc}
\usepackage{amsmath}
\usepackage{amsfonts}
\usepackage{amssymb,comment,xspace}
\usepackage[ruled,vlined,linesnumbered]{algorithm2e}
\usepackage{graphicx}
\usepackage{pgf}
\usepackage{tikz}
\usetikzlibrary{arrows,automata,positioning}

\usepackage{multirow}

\usepackage{url}

\newcommand{\R}{\mathcal{R}}

\renewcommand{\O}{\mathcal{O}}
\newcommand{\F}{\mathcal{F}}

\newcommand{\InsertSort}{\textsc{InsertionSort}\xspace}

\newcommand{\EE}{\mathbb{E}}
\newcommand{\PP}{\mathbb{P}}

\newcommand{\wqotk}[2]{\weight^{\quotient}_{#1}(#2)}

\DeclareMathOperator{\sym}{\mathfrak{S}}
\DeclareMathOperator{\weight}{w}
\DeclareMathOperator{\quotient}{q}
\DeclareMathOperator{\Weight}{W}
\DeclareMathOperator{\cyc}{cycle}
\DeclareMathOperator{\rec}{record}

\DeclareMathOperator{\norm}{norm}
\DeclareMathOperator{\rank}{rank}
\DeclareMathOperator{\desc}{desc}
\DeclareMathOperator{\inv}{inv}
\DeclareMathOperator{\below}{below}

\newcommand{\rfact}[2]{#1^{(#2)}}
\newcommand{\symin}[2]{\sym_{#1 \text{\,in\,} #2}}   

\newtheorem{theorem}{Theorem}
\newtheorem{lemma}[theorem]{Lemma}
\newtheorem{corollary}[theorem]{Corollary}

\newenvironment{remark}{\noindent{\bf Remark:} }{\hfill$\diamond$}

\author{Nicolas Auger\addressmark{1}
  \and Mathilde Bouvel\addressmark{2}\thanks{Supported by a Marie Heim-Vögtlin grant of the Swiss National Science Foundation.}
  \and Cyril Nicaud\addressmark{1}
  \and Carine Pivoteau\addressmark{1}}

\title[Analysis of Algorithms for Biased Permutations]{Analysis of Algorithms for Permutations Biased by Their Number of Records}

\address{
\addressmark{1} Universit\'e Paris-Est, LIGM (UMR 8049), F77454 Marne-la-Vall\'ee, France\\
\addressmark{2} CNRS and Institut f\"ur Mathematik, Universit\"at Z\"urich, Zurich, Switzerland}

\keywords{permutation, Ewens distribution, random generation, analysis of algorithms}

\begin{document}
\maketitle
\begin{abstract}
The topic of the article is the parametric study of the complexity of algorithms on arrays of pairwise distinct integers. 
We introduce a model that takes into account the non-uniformness of data, 
which we call the Ewens-like distribution of parameter $\theta$ for records on permutations: 
the weight $\theta^r$ of a permutation depends on its number $r$ of records. 
We show that this model is  meaningful for the notion of presortedness, while still being mathematically tractable.
Our results describe the expected value of several classical permutation statistics in this model, 
and give the expected running time of three algorithms: the Insertion Sort, and two variants of the Min-Max search. 
\end{abstract}


\section{Introduction}\label{sec:intro}
A classical framework for analyzing the average running time of algorithms is to consider uniformly distributed inputs. 
Studying the complexity of an algorithm under this uniform model usually gives a quite good understanding of the algorithm. 
However, it is not always easy to argue that the uniform model is relevant, when the algorithm is used on a specific data set. 
Observe that, in some situations, the uniform distribution arises by construction, from the  randomization of a deterministic algorithm. 
This is the case with Quick Sort for instance, when the pivot is chosen uniformly at random.
In other situations, the uniformity assumption may not fit the data very well, 
but still is a reasonable first step in modeling it, which  makes the analysis mathematically tractable.

In practical applications where the data is a sequence of values, 
it is not unusual that the input is already partially sorted, depending on its origin. 
Consequently, assuming that the input is uniformly distributed, or shuffling the input as in the case of Quick Sort, may not be a good idea.  
Indeed, in the last decade, standard libraries of well-established languages have switched to sorting algorithms that take advantage of the ``almost-sortedness'' of the input. 
A noticeable example is Tim Sort Algorithm, used in Python (since 2002) and Java (since Java 7): it is particularly efficient to process data consisting of long increasing (or decreasing) subsequences.

In the case of sorting algorithms, the idea of taking advantage of some bias in the data towards sorted sequences dates back to Knuth~\cite[p. 336]{Knuth98}. 
It has been embodied by the notion of \emph{presortedness}, which quantifies how far from sorted a sequence is. 
There are many ways of defining measures of presortedness, and it has been axiomatized by Mannila~\cite{Mannila1985} (see Section~\ref{sec:presortedness} for a brief overview). 
For a given measure of presortedness $m$, the classical question is to find a sorting algorithm that is optimal for $m$, meaning that it minimizes the number of comparisons as a function of both the size of the input and the value of~$m$. 
For instance, Knuth's Natural Merge Sort~\cite{Knuth98} is optimal for the measure $r=$``number of runs'' , with a worst case running time of $\O(n\log r)$ for an array of length $n$. 

Most measures of presortedness studied in the literature are directly related to basic statistics on permutations. 
Consequently, it is natural to define biased distributions on permutations that depend on such statistics, 
and to analyze classical algorithms under these non-uniform models. 
One such distribution is very popular in the field of discrete probability: the Ewens distribution. 
It gives to each permutation~$\sigma$ a probability that is proportional to $\theta^{\cyc(\sigma)}$,
where $\theta>0$ is a parameter and $\cyc(\sigma)$ is the number of cycles in $\sigma$. 
Similarly, for any classical permutation statistics $\chi$, 
a non-uniform distribution on permutations may be defined by giving to any $\sigma$ a probability proportional to $\theta^{\chi(\sigma)}$. 
We call such distributions \emph{Ewens-like distributions}. 
Note that the Ewens-like distribution for the number of inversions is quite popular, under the name of \emph{Mallows distribution}~\cite[and references therein]{Gladkich}.

In this article, we focus on the Ewens-like distribution according to $\chi=$ number of records (a.k.a. \emph{left to right maxima}). 
The motivation for this choice is twofold. 
First, the number of records is directly linked to the number of cycles by the fundamental bijection (see Section~\ref{sec:permutations}). 
So, we are able to exploit the nice properties of the classical Ewens distribution, and have a non-uniform model that remains mathematically tractable. 
Second, we observe that the number of non-records is a measure of presortedness. 
Therefore, our distribution provides a model for analyzing algorithms which is meaningful for the notion of presortedness, 
and consequently which may be more realistic than the uniform distribution. 
We first study how this distribution impacts the expected value of some classical permutation statistics, depending on the choice of~$\theta$. 
Letting $\theta$ depend on $n$, we can reach different kinds of behavior. 
Then, we analyze the expected complexity of Insertion Sort under this biased distribution, 
as well as the effect of branch prediction on two variants of the simultaneous minimum and maximum search in an array.

\section{Permutations and Ewens-like distributions}\label{sec:def}

\subsection{Permutations as words or sets of cycles}\label{sec:permutations}

For any integers~$a$ and~$b$, let $[a,b]=\{a,\ldots,b\}$ and for every integer $n\geq 1$, let $[n]=[1,n]$. 
By convention $[0]=\emptyset$. 
If $E$ is a finite set, let $\sym(E)$ denote the set of all permutations on $E$, \emph{i.e.}, of bijective maps from $E$ to itself. 
For convenience, $\sym([n])$ is written $\sym_{n}$ in the sequel. 
Permutations of $\sym_{n}$ can be seen in several ways (reviewed for instance in~\cite{Bona}). 
Here, we use both their representations as words and as sets of cycles. 

A permutation $\sigma$ of $\sym_{n}$ can be represented as a word $w_1 w_2 \ldots w_n$ containing exactly once each symbol in $[n]$: 
by simply setting $w_i = \sigma(i)$ for all $i \in [n]$. 
Conversely, any sequence (or word) of $n$ distinct integers can be interpreted as representing a permutation of $\sym_{n}$. 
For any sequence $s = s_1 s_2  \ldots s_n$ of~$n$ distinct integers, the rank $\rank_{s}(s_i)$ of $s_i$ is defined as
the number of integers appearing in $s$ that are smaller than or equal to $s_i$.
Then, for any sequence $s$ of $n$ distinct integers, the \emph{normalization} $\norm(s)$ of~$s$ is the
unique permutation $\sigma$ of $\sym_{n}$ such that $\sigma(i)=\rank_{s}(s_i)$. For instance, $\norm(8254) = 4132$.

Many permutation statistics are naturally expressed on their representation as words. 
One will be of particular interest for us: the number of records. 
If $\sigma$ is a permutation of $\sym_{n}$ and $i\in[n]$,
there is a \emph{record} at position $i$ in $\sigma$ (and subsequently, $\sigma(i)$ is a record) if $\sigma(i)>\sigma(j)$ for every $j\in [i-1]$. 
In the word representation of permutations, records are therefore elements that have no larger elements to their left. 
This equivalent definition of records naturally extends to sequences of distinct integers, and 
for any sequence $s$ of distinct integers, the positions of the records in $s$ and in $\norm(s)$ are the same. 
A position that is not a record is called a \emph{non-record}.

A cycle of size~$k$ in a permutation $\sigma \in \sym_{n}$ is a subset $\{i_1,\dots,i_k\}$ of $[n]$ such that 
$i_1 \stackrel{\sigma}{\mapsto} i_2 \ldots \stackrel{\sigma}{\mapsto} i_k \stackrel{\sigma}{\mapsto} i_1$. 
It is written $(i_1, i_2, \ldots, i_k)$. 
Any permutation can be decomposed as the set of its cycles. 
For instance, the cycle decomposition of $\tau$ represented by the word $6321745$ is $(32) (641) (75)$. 

These two ways of looking at permutations (as words or as set of cycles) are rather orthogonal, 
but there is still a link between them, provided by the so-called \emph{fundamental bijection} or \emph{transition lemma}. 
The fundamental bijection, denoted $F$, is the following transformation: 
\begin{enumerate}
\setlength\itemsep{-1mm}
\item Given $\sigma$ a permutation of size $n$, consider the cycle decomposition of $\sigma$. 
\item Write every cycle starting with its maximal element, 
and write the cycles in increasing order of their maximal (i.e., first) element. 
\item Erasing the parenthesis gives $F(\sigma)$. 
\end{enumerate}
Continuing our previous example gives $F\big( \tau \big) = 3264175$. 
This transformation is a bijection, and transforms a permutation as set of cycles into a permutation as word. 
Moreover, it maps the number of cycles to the number of records. 
For references and details about this bijection, see for example~\cite[p. 109--110]{Bona}. 

\subsection{The number of non-records as a measure of presortedness}\label{sec:presortedness}

The concept of presortedness, formalized by Mannila~\cite{Mannila1985}, naturally arises when studying sorting algorithms which efficiently sort sequences already almost sorted. 
Let $E$ be a totally ordered set. 
We denote by $E^{\star}$ the set of all nonempty sequences of distinct elements of $E$, 
and by $\cdot$ the concatenation on $E^{\star}$. 
A mapping $m$ from  $E^{\star}$ to $\mathbb N$ is a \emph{measure of presortedness} if it satisfies:
\begin{enumerate}
\setlength\itemsep{-1mm}
\item if $X\in E^{\star}$ is sorted then $m(X)=0$;
\item if $X=(x_{1},\cdots,x_{\ell})$ and $Y=(y_{1},\cdots,y_{\ell})$ are two elements of $E^{\star}$ having same length, and such that for every $i,j\in[\ell]$,
$x_{i}<x_{j}\Leftrightarrow y_{i}<y_{j}$ then $m(X)=m(Y)$;
\item if $X$ is a subsequence of $Y$ then $m(X)\leq m(Y)$;
\item if every element of $X$ is smaller than every element of $Y$ then $m(X\cdot Y)\leq m(X) + m(Y)$;
\item for every symbol $a\in E$ that does not occur in $X$, $m(a\cdot X)\leq |X|+m(X)$.
\end{enumerate}
Classical measures of presortedness~\cite{Mannila1985} are the number of inversions, the number of swaps, \ldots 
One can easily see, checking conditions 1 to 5, that $m_{\text{rec}}(s) =$ number of non-records in $s = |s| -$ number of records in $s$ 
defines a measure of presortedness on sequences of distinct integers. 
Note that because of condition 2, studying a measure of presortedness on $\sym_n$ is not a restriction with respect to studying it on sequences of distinct integers. 

Given a measure of presortedness $m$, we are interested in optimal sorting algorithms with respect to $m$. 
Let $\below_m(n,k)=\{\sigma:\sigma\in\sym_{n},\ m(\sigma)\leq k\}$. 
A sorting algorithm is \emph{$m$-optimal} (see~\cite{Mannila1985} and~\cite{Petersson95} for more details) 
if it performs in the worst case $\O(n+\log|\below_m(n,k)|)$ comparisons when applied to $\sigma \in \sym_n$ such that $m(\sigma)=k$, uniformly in $k$. 
There is a straightforward algorithm  that is $m_{\text{rec}}$-optimal. First scan $\sigma$ from left to right and put the records in one (sorted) list $L_{R}$ and the non-records in another list $L_{N}$. Sort $L_{N}$ using a $\O(|L_{N}|\log |L_{N}|)$ algorithm, then merge it with
$L_{R}$. The worst case running time of this algorithm is $\O(n+k\log k)$ for 
permutations $\sigma$ of $\sym_{n}$ such that $m_{\text{rec}}(\sigma)=k$.
Moreover, $|\below_{m_{\text{rec}}}(n,k)| \geq k!$ for any $k \geq n$, 
since it contains the $k!$ permutations of the form $(k+1) (k+2)\ldots n \cdot \tau$ for $\tau\in\sym_{k}$. 
Consequently, $\O(n+k\log k) = \O(n+\log|\below_{m_{\text{rec}}}(n,k)|))$, proving $m_{\text{rec}}$-optimality. 

\subsection{Ewens and Ewens-like distribution}

The Ewens distribution on permutations (see for instance~\cite[Ch. 4 \& 5]{Arratia}) is a generalization of the uniform distribution on $\sym_{n}$: 
the probability of a permutation depends on its number of cycles. 
Denoting $\cyc(\sigma)$ the number of cycles of any permutation $\sigma$, 
the Ewens distribution of parameter $\theta$ (where $\theta$ is any fixed positive real number) 
gives to any $\sigma$ the probability $\frac{\theta^{\cyc(\sigma)}}{\sum_{\rho\in\sym_{n}} \theta^{\cyc(\rho)}}$. 
As seen in~\cite[Ch. 5]{Arratia}, the normalization constant $\sum_{\rho\in\sym_{n}} \theta^{\cyc(\rho)}$ is $\rfact\theta n$, 
where the notation $\rfact{x}{n}$ (for any real $x$) denotes the \emph{rising factorial} defined by
$\rfact{x}{n}=x(x + 1)\cdots(x + n-1)$ (with the convention that $\rfact{x}{0}=1$).

Mimicking the Ewens distribution, it is natural (and has appeared on several occasions in the literature, see for instance~\cite[Example 12]{Borodin}) 
to define other non-uniform distributions on $\sym_{n}$, 
where we introduce a bias according to some statistics $\chi$. 
The Ewens-like distribution of parameter $\theta$ (again $\theta$ is any fixed positive real number) 
for statistics $\chi$ 
is then the one that gives to any $\sigma \in \sym_n$ the probability $\frac{\theta^{\chi(\sigma)}}{\sum_{\rho\in\sym_{n}} \theta^{\chi(\rho)}}$. 
The classical Ewens distribution corresponds to $\chi = $ number of cycles. 
Ewens-like distributions can be considered for many permutations statistics, like the number of inversions, of fixed points, of runs,~\dots 
In this article, we focus on the distribution associated with $\chi = $ number of \emph{records}. 
We refer to it as the Ewens-like distribution for records (with parameter $\theta$).
For any $\sigma$, we let $\rec(\sigma)$ denote the number of records of $\sigma$, and define the \emph{weight} of $\sigma$ as $\weight(\sigma)=\theta^{\rec(\sigma)}$. 
The Ewens-like distribution for records on $\sym_{n}$ gives probability $\frac{\weight(\sigma)}{\Weight_{n}}$ to any $\sigma \in \sym_{n}$, 
where $\Weight_{n}=\sum_{\rho\in\sym_{n}} \weight(\rho)$. 
Note that the normalization constant is $\Weight_{n} =\rfact\theta n $, like in the classical Ewens distribution: 
indeed, the fundamental bijection reviewed above 
shows that there are as many permutations with $c$ cycles as permutations with $c$ records. 
Fig.~\ref{fig:Ewens} shows random permutations under the Ewens-like distribution for records, for various values of $\theta$.

\begin{figure}[ht]
\includegraphics[scale=.18]{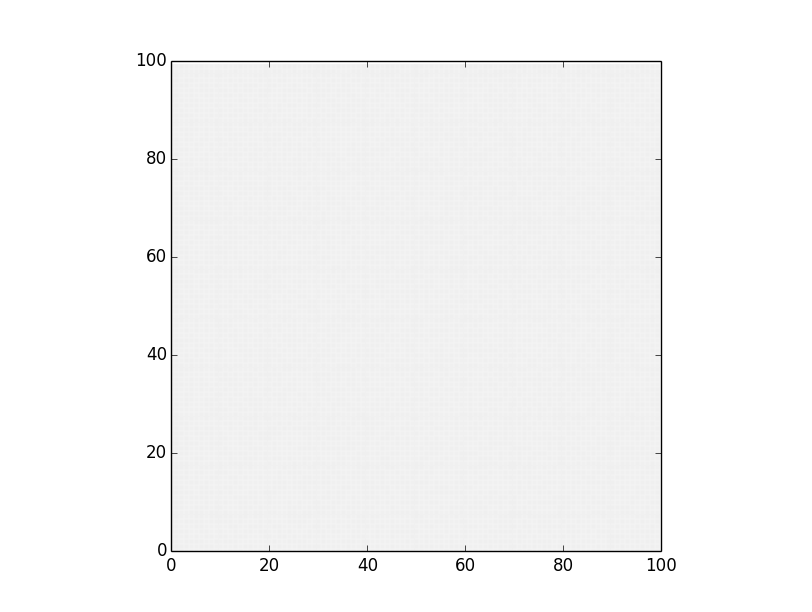}\includegraphics[scale=.18]{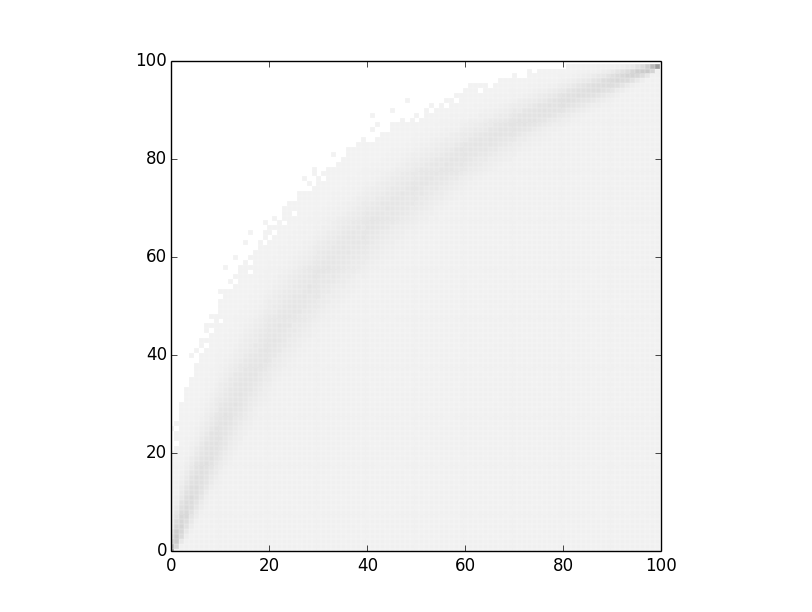}\includegraphics[scale=.18]{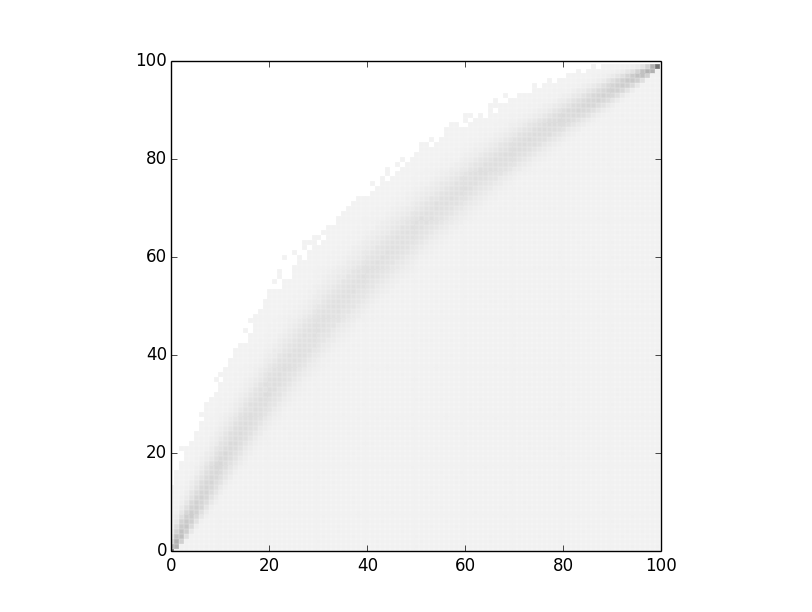}\includegraphics[scale=.18]{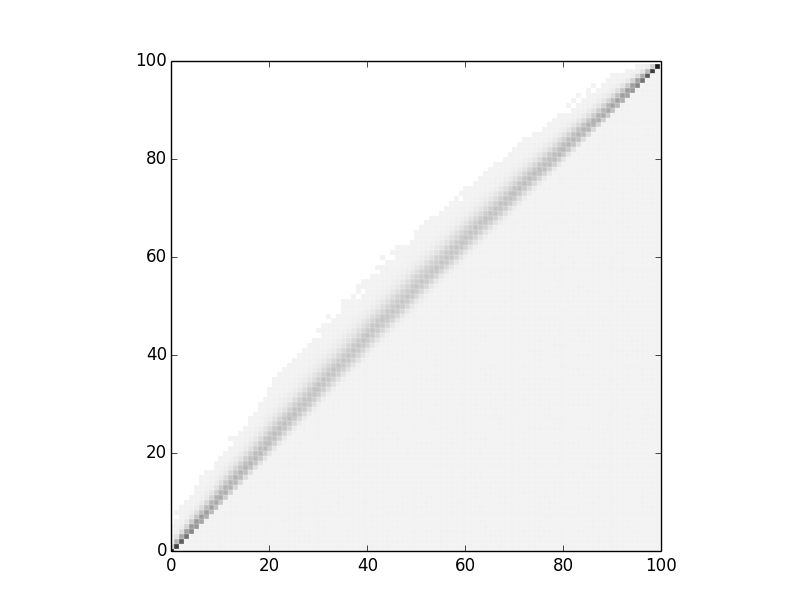}
\caption{Random permutations under the Ewens-like distribution on $\sym_{100}$ with, from left to right, $\theta=1$ (corresponding to the uniform distribution)$,50,100,$ and $500$. For each diagram,  the darkness of a point $(i,j)$ is proportional 
to the number of generated permutations $\sigma$ such that $\sigma(i)=j$, for a sampling of 10000 random permutations.\label{fig:Ewens}}
\end{figure}

\subsection{Linear random samplers}
Efficient random samplers have several uses for the analysis of algorithms in general. 
They allow to estimate quantities of interest (even when their computation with a theoretical approach is not feasible), 
and can be used to double-check theoretical results. 
They are also a precious tool to visualize the objects under study (the diagrams in Fig.~\ref{fig:compare} were obtained in this way), 
allowing to define new problems on these objects (for example: can we describe the limit shape of the diagrams shown in Fig.~\ref{fig:compare}?).

As mentioned in~\cite[\S 2.1]{Valentin}, one can easily obtain a linear time and space algorithm to generate a random permutation according to the Ewens distribution (for cycles), using a variant of the Chinese restaurant process reviewed in what follows. 
To generate a permutation of size~$n$, we start with an empty array\footnote{Note that our array starts at index~1.}~$\sigma$ of length~$n$ that is used to store the values of the~$\sigma(i)$'s. For~$i$ from~1 to~$n$, we choose to either create a new cycle containing only~$i$ with probability~$\frac{\theta}{\theta + i-1}$ or to insert $i$ in one of the existing cycles with probability~$\frac{i-1}{\theta + i-1}$. To create a new cycle, we set $\sigma[i]=i$. To insert~$i$ in an existing cycle, we choose uniformly at random an element~$j$ in~$[i-1]$ to be the element following~$i$ in its cycle, and we set $\sigma[i]=j$ and $\sigma[\sigma^{-1}[j]]=i$. To avoid searching for~$\sigma^{-1}[j]$ in the array $\sigma$, we only need to keep $\sigma^{-1}$ in a second array while adding the elements in $\sigma$.

Starting from this algorithm, we can easily design a linear random sampler for permutations according to the Ewens-like distribution for records, using the fundamental bijection. The first step is to generate a permutation $\sigma$ in~$\sym_n$ with the above algorithm. Then, we write the cycles of~$\sigma$ in reverse order of their maximum, as sequences, starting from the last element and up to exhaustion of the cycle: $n, \sigma[n], \sigma[\sigma[n]], \dots, \sigma^{-1}[n]$. Each time we write an element~$i$, we set~$\sigma[i]=0$ and each time a cycle is finished, we search the next value of~$i$ such that $\sigma[i]\neq 0$ to start the next cycle. This new cycle will be written before the one that has just been written. Note that all these operations can be performed in time complexity~$O(1)$ using doubly linked lists for the resulting permutation. In the end, the cycles will be written as sequences starting by their maximum, sorted in increasing order of their maximum, which is the fundamental bijection.

Note that there exists another branching process, known as the {\em Feller coupling}, to generate permutations according to the Ewens distribution (see for instance~\cite[p.16]{Arratia}). 
Although it is less natural than with the Chinese restaurant process, it is also possible to infer linear random samplers from it. 
Details will be provided in an extended version of this work. 

\section{Average value of statistics in biased random permutations}\label{sec:stats}

Let $\theta$ be any fixed positive real number. 
In this section, we study the behavior of several statistics on permutations, 
when they follow the Ewens-like distribution for records with parameter $\theta$. 
Our purpose is mostly to illustrate methods to obtain precise descriptions of the behavior of such statistics. 
Such results allow a fine analysis of algorithms whose complexity depends on the studied statistics.

Recall that, for any $\sigma\in\sym_{n}$, $\weight(\sigma)=\theta^{\rec(\sigma)}$ 
and the probability of $\sigma$ is $\frac{\weight(\sigma)}{\Weight_n}$, with $\Weight_n =\rfact{\theta}{n}$. 
Recall also that the records of any sequence of distinct integers are well-defined.  
For any such sequence $s$ we subsequently set 
$\rec(s)$ to be the number of records of $s$ and $\weight(s) = \theta^{\rec(s)}$. 
Note that for any such sequence $s$, $\weight(s) = \weight(\norm(s))$, 
because the positions (and hence the number) of records do not change when normalizing. 

\subsection{Technical lemmas}

Seeing a permutation of $\sym_{n}$ as a word, it can be split (in many ways) 
into two words as $\sigma = \pi\cdot\tau$ for the usual concatenation on words. 
Note that here~$\pi$ and~$\tau$ are not normalized permutations:  $\tau$ belongs to the set~$\symin{k}{n}$ of all sequences of $k$ distinct integers in $[n]$ where~$k=|\tau|$, and $\pi$ belongs to~$\symin{n-k}{n}$.
The weight function $\weight$ behaves well with respect to this decomposition, as shown in the following lemmas. 

\begin{lemma}\label{lem:weight prime}
Let $n$ be an integer, and $\tau$ be a sequence of $k \leq n$ distinct integers in $[n]$. 
Denote by $m$ the number of records in $\tau$ whose value is larger than the largest element of $[n]$ which does not appear in $\tau$,
and define $\weight'_{n}(\tau)$ as $\theta^{m}$. 
For all $\sigma \in \sym_n$, if $\sigma = \pi\cdot\tau$, then $\weight(\sigma) = \weight(\pi) \cdot \weight'_{n}(\tau)$.
\end{lemma}

For instance, the definition of $\weight'_{n}(\tau)$ gives $\weight'_{9}(6489)= \theta^2$ ($8$ and $9$ are records of $\tau$ larger than $7$) 
and $\weight'_{10}(6489)= 1$ (there are no records in $\tau$ larger than $10$). 

We extend the weight function $\weight$ to subsets $X$ of $\sym_{n}$ as $\weight(X)=\sum_{\sigma\in X}\weight(\sigma)$. 
For any sequence $\tau$ of $k \leq n$ distinct integers in $[n]$, the right-quotient of $X$ with $\tau$ is $X/\tau = \{\pi:\pi\cdot\tau\in X\}$. 
Since $\weight(\pi) = \weight(\norm(\pi))$ for all sequences $\pi$ of distinct integers, 
we have $\weight(X/\tau) = \weight(\norm(X/\tau))$ for all $X$ and $\tau$ as above 
(As expected, $\norm(Y)$ means $\{\norm(\pi) : \pi \in Y\}$).


For $k\in[n]$, we say that $X \subseteq \sym_{n}$ is \emph{quotient-stable for $k$} if  
$\weight(X/\tau)$ is constant when $\tau$ runs over $\symin{k}n$. 
When $X$ is quotient-stable for $k$, we denote $\wqotk{k}{X}$ the common value of $\weight(X/\tau)$ for $\tau$ as above. 
For instance, $X=(4321,3421,4132,3142,4123,2143,3124,1324)$ is quotient-stable for $k=1$. Indeed, 
\begin{align*}
\wqotk{1}{X}=\weight(X/1)=&\weight(\{432,342\})=
\weight(X/2)=\weight(\{413,314\})=\\
\weight(X/3)=&\weight(\{412,214\})=
\weight(X/4)=\weight(\{312,132\})=\theta+\theta^2.
\end{align*}

Note that $\sym_n$ is quotient-stable for all $k  \in [n]$: 
indeed, for any $\tau$ of size $k$, $\norm(\sym_n/\tau) = \sym_{n-k}$ so that $\weight(\sym_n/\tau) = \weight(\sym_{n-k})$ for all $\tau$ of size $k$. 
It follows that $\wqotk{k}{\sym_n}= \weight(\sym_{n-k}) = \rfact \theta {n-k}$. 

\begin{lemma}\label{lem:qot}
Let $X\subseteq \sym_{n}$ be quotient-stable for $k \in [n]$. Then 
$\weight(X) = \frac{\rfact \theta n}{\rfact \theta {n-k}}\wqotk{k}{X}$.
\end{lemma}

A typical example of use of~Lemma~\ref{lem:qot} is given in the proof of Theorem~\ref{lem:records}.

\begin{remark}
Lemma~\ref{lem:qot} is a combinatorial version of a simple probabilistic property: Let $E_{\tau}$ be the set of elements of $\sym_{n}$ that end with $\tau$. If $A$ is an event on $\sym_{n}$ and if the probability of $A$ given $E_{\tau}$ is the same for every $\tau\in\symin{k}{n}$, then it is equal to the probability of $A$, by the law of total probabilities.
\end{remark}

\subsection{Summary of asymptotic results}

The rest of this section is devoted to studying the expected behavior of some permutation statistics, 
under the Ewens-like distribution on $\sym_n$ for records with parameter $\theta$. 
We are especially interested in the asymptotics in $n$ when $\theta$ is constant or is a function of $n$. 
The studied statistics are: number of records, number of descents, first value, and number of inversions. 
A summary of our results is presented in Table~\ref{table:asymptotics}. 
The asymptotics reported in Table~\ref{table:asymptotics} follow from Corollaries~\ref{thm:nb_of_records}, \ref{thm:nb_of_descents}, \ref{thm:first_value}, \ref{thm:nb_of_inversions} 
either immediately or using the so-called \emph{digamma} function. 
The digamma\footnote{For details, see \url{https://en.wikipedia.org/wiki/Digamma_function} (accessed on April 27, 2016).} function is defined by $\Psi(x)=\Gamma'(x)/\Gamma(x)$. 
It satisfies the identity $\sum_{i=0}^{n-1}\frac1{x + i} = \Psi(x + n)-\Psi(x)$, 
and its asymptotic behavior as $x\rightarrow\infty$ is $\Psi(x) = \log(x) - \frac{1}{2x} -\frac{1}{12x^2}  +  o\left(\frac{1}{x^2}\right)$. 
We also define $\Delta(x,y)=\Psi(x+y)-\Psi(x)$, so that $\Delta(x,n)=\sum_{i=0}^{n-1}\frac1{x + i}$ for any positive integer $n$.
In Table~\ref{table:asymptotics} and in the sequel,
we use the notations $\PP_{n}(E)$ (resp. $\EE_{n}[\chi]$) to denote 
the probability of an event $E$ (resp. the expected value of a statistics $\chi$)   
under the Ewens-like distribution on $\sym_n$ for records.
 
\begin{table}[ht]\label{table:asymptotics}
\begin{center}
\begin{tabular}{l|l|l|l|l|l|c}
& $\theta=1$ & fixed $\theta>0$ & $\theta := n^{\epsilon}$,  & $\theta :=  \lambda n$, & $\theta := n^{\delta}$, &See \\
& {\small (uniform)}& & {\small $0<\epsilon<1$ }& {\small $\lambda>0$} & {\small $\delta>1$ }&Cor.\\
\hline 
$\EE_{n}[\rec]$ & $ \log n $ & $ \theta\cdot \log n $ & $ (1-\epsilon)\cdot n^{\epsilon}\log n$ & $ \lambda \log(1 + 1/\lambda)\cdot n$ & $ n$ & \ref{thm:nb_of_records}\\
$\EE_{n}[\desc]$ & $  n/2 $ & $  n/2 $ & $  n/2$ & $ n/2(\lambda + 1)$ & $ n^{2-\delta}/2$ & \ref{thm:nb_of_descents}\\
$\EE_{n}[\sigma(1)]$ & $ n/2$ & $  n/(\theta  + 1) $ & $  n^{1-\epsilon}$ & $ (\lambda + 1)/\lambda$ & $ 1$ & \ref{thm:first_value}\\
$\EE_{n}[\inv]$ & $ n^2/4$ & $ n^2/4$ & $ {n^{2}}/4$ & $ {n^{2}}/4 \cdot f(\lambda)$ & $ {n^{3-\delta}}/{6}$ & \ref{thm:nb_of_inversions}\\
\end{tabular}
\end{center}
\caption{Asymptotic behavior of some permutation statistics under the Ewens-like distribution on $\sym_n$ for records.  
We use the shorthand $f(\lambda) = 1-2\lambda + 2\lambda^{2}\log\left(1 + 1/\lambda\right)$.
All the results in this table are asymptotic equivalents.
}
\end{table}

\begin{remark}
To some extent, our results may also be interpreted on the classical Ewens distribution, via the fundamental bijection. 
Indeed the number of records (resp. the number of descents, resp. the first value) of $\sigma$ 
corresponds to the number of cycles (resp. the number of anti-excedances\footnote{An anti-excedance of $\sigma \in \sym_n$ is $i \in [n]$ such that $\sigma(i) <i$. 
The proof that descents of $\sigma$ are equinumerous with anti-excedances of $F^{-1}(\sigma)$ is a simple adaptation of the proof of Theorem 1.36 in~\cite{Bona}, p. 110--111.}
, resp. the minimum over all cycles of the maximum value in a cycle) 
of $F^{-1}(\sigma)$. 
Consequently, Corollary~\ref{thm:nb_of_records} is just a consequence of the well-known expectation 
of the number of cycles under the Ewens distribution (see for instance~\cite[\S 5.2]{Arratia}).  
Similarly, the expected number of anti-excedances (Corollary~\ref{thm:nb_of_descents}) can be derived easily from the results of~\cite{Valentin}. 
Those results on the Ewens distribution do not however give access to results as precise as those stated in Theorems~\ref{lem:records} and~\ref{lem:descent i}, 
which are needed to prove our results of Section~\ref{sec:mispredictions}. 
Finally, to the best of our knowledge, the behavior of the third statistics (minimum over all cycles of the maximum value in a cycle) has not been previously studied, 
and we are not aware of any natural interpretation of the number of inversions of $\sigma$ in $F^{-1}(\sigma)$. 
\end{remark}

\subsection{Expected values of some permutation statistics}

We start our study by computing how the value of parameter $\theta$ influences the expected number of records.

\begin{theorem}\label{lem:records}
Under the Ewens-like distribution on $\sym_n$ for records with parameter $\theta$, 
for any $i \in [n]$, 
the probability that there is a record at position $i$ is: 
$\PP_{n}(\text{record at }i)  = \frac{\theta}{\theta  +  i - 1}$.
\end{theorem}

\begin{proof}
We prove this theorem by splitting permutations seen as words after their $i$-th element, as shown in Fig.~\ref{fig:record}.
Let $\R_{n,i}$ denote the set of permutations of $\sym_{n}$ having a record at position $i$. 
We claim that the set~$\R_{n,i}$ is quotient-stable for $n-i$, and that $\weight(\R_{n,i}) = \frac{\rfact{\theta}{n}}{\rfact{\theta}{i}} \cdot \rfact{\theta}{i-1} \cdot \theta$. 
It will immediately follow that $\PP_{n}(\text{record at }i) = 
\frac{\weight(\R_{n,i})}{\rfact{\theta}{n}} = \frac{\rfact{\theta}{i-1} \cdot \theta}{\rfact{\theta}{i}} = \frac{\theta}{\theta  + i -1}$.
We now prove the claim. 
Let $\tau$ be any sequence in $\symin{n-i}{n}$. 
Observe that $\norm(\R_{n,i}/\tau) = \R_{i,i}$. 
Since the number of records is stable by normalization, it follows that $\weight(\R_{n,i}/\tau)=\weight(\R_{i,i})$. 
By definition, $\pi \in \sym_i$ is in $\R_{i,i}$ if and only if $\pi(i)=i$. Thus $\R_{i,i}=\sym_{i-1}\cdot\, i$ in the word representation of permutations. 
Hence, $\weight(\R_{i,i})={\rfact\theta {i-1}}\theta$, since the last element is a record by definition. 
This yields $\weight(\R_{n,i}/\tau) = {\rfact\theta {i-1}}\theta$ for any $\tau \in \symin{n-i}{n}$, 
proving that $\R_{n,i}$ is quotient-stable for $n-i$, 
and that $\wqotk{n-i}{\R_{n,i}} = {\rfact\theta {i-1}}\theta$. 
 By Lemma~\ref{lem:qot}, it follows that 
$\weight(\R_{n,i}) = \frac{\rfact\theta {n}}{\rfact\theta {n-(n-i)}} \cdot \wqotk{n-i}{\R_{n,i}}= 
\frac{\rfact{\theta}{n}}{\rfact{\theta}{i}} \cdot \rfact{\theta}{i-1} \cdot \theta$.
\end{proof}

\begin{figure}[ht]
\begin{minipage}[]{.45\textwidth}
\includegraphics[scale=.32]{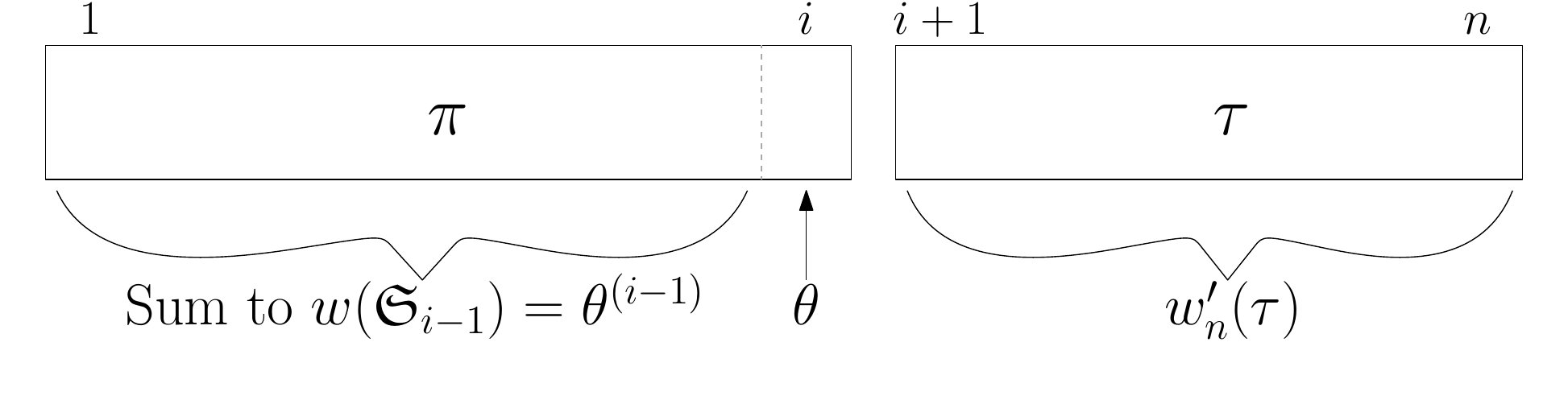}
\end{minipage}
\begin{minipage}[]{.54\textwidth}
\caption{The decomposition used to compute the probability of having a record at $i$. 
This record has weight $\theta$ and thus, for any fixed $\tau$, the weights of all possible $\pi$ sum to $\weight(\sym_{i-1})\cdot\theta=\rfact\theta{i-1}\cdot\theta$.\label{fig:record}}
\end{minipage}
\end{figure}

\begin{corollary} \label{thm:nb_of_records}
Under the Ewens-like distribution on $\sym_n$ for records with parameter $\theta$, the expected value of the number of records is: 
$\EE_{n}[\rec] = \sum_{i=1}^{n}\frac{\theta}{\theta + i-1} = \theta \cdot\Delta(\theta,n)$.
\end{corollary}

\medskip

Next, we study the expected number of descents. 
Recall that a permutation $\sigma$ of $\sym_{n}$ has a \emph{descent} at position $i\in\{2,\ldots,n\}$ if
$\sigma(i-1)>\sigma(i)$. We denote by $\desc(\sigma)$ the number of descents in $\sigma$. 
We are interested in descents as they are directly related to the number of increasing runs in a permutation 
(each such run but the last one is immediately followed by a descent, and conversely). 
Some sorting algorithms, like Knuth's Natural Merge Sort, use the decomposition into runs.

The following theorem is proved using Lemmas~\ref{lem:weight prime} and~\ref{lem:qot} and the decomposition of Fig.~\ref{fig:desc}.
\begin{theorem}\label{lem:descent i}
Under the Ewens-like distribution on $\sym_n$ for records with parameter $\theta$, 
for any $i\in\{2,\ldots,n\}$, 
the probability that there is a descent at position $i$ is: 
$\PP_{n}\big(\sigma(i-1)>\sigma(i)\big) =   \frac{(i-1)(2\theta  +  i-2)}{2(\theta + i-1)(\theta + i-2)}$. 
\end{theorem}
\begin{figure}[ht]
\begin{minipage}{.49\textwidth}
\includegraphics[scale=.35]{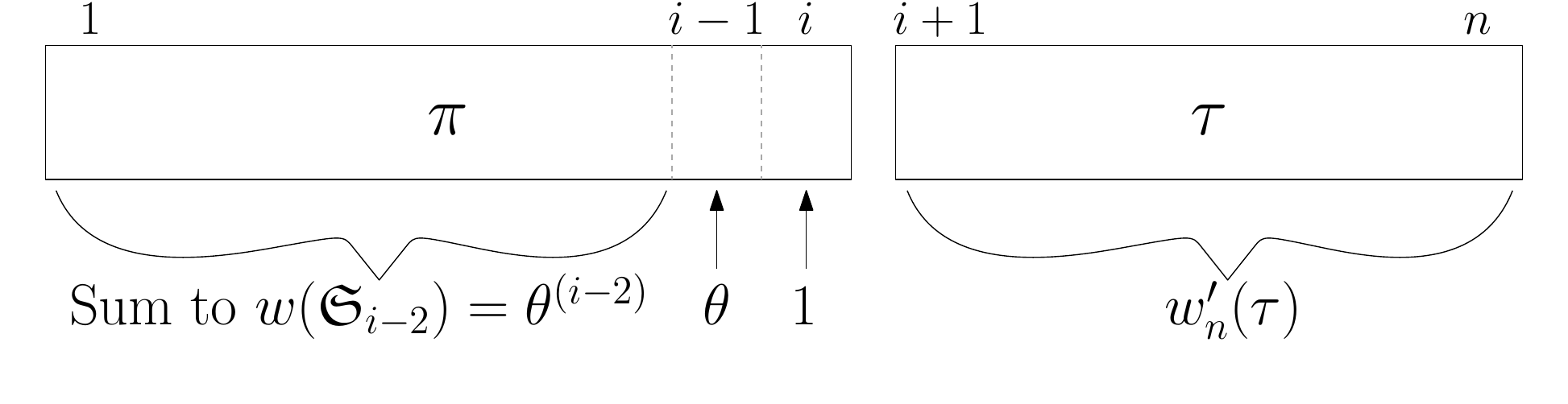}
\end{minipage}
\begin{minipage}{.49\textwidth}
\includegraphics[scale=.35]{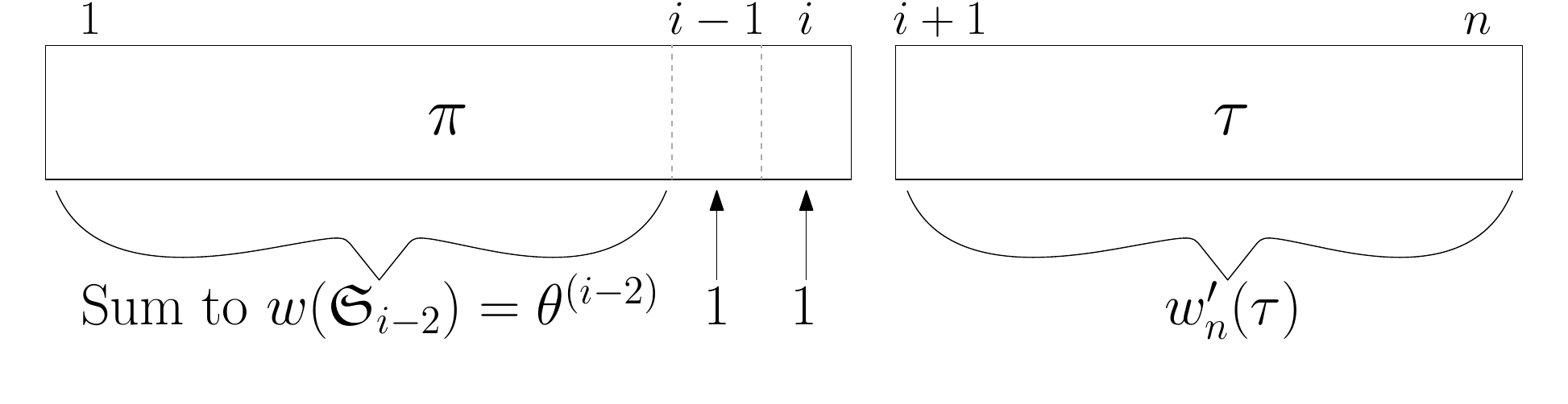}
\end{minipage}
\caption{The two cases for the probability of having a descent at $i$. 
We decompose $\sigma$ as $\pi \cdot \sigma(i-1) \cdot \sigma(i) \cdot \tau$, and we let $\rho=\norm(\pi \cdot \sigma(i-1) \cdot \sigma(i))$.
On the left, the case where $\sigma(i-1)$ is a record, that is, $\rho(i-1)=i$: there are $i-1$ possibilities for $\rho(i)$. 
On the right, the case where $\sigma(i-1)$ is not a record: there are $\binom{i-1}{2}$ possibilities for the values of $\rho(i)$ and $\rho(i-1)$. 
In both cases, once the images of $j\in\{i-1,\ldots n\}$ by $\sigma$ have been chosen, the weight of all possible beginnings sum to $w(\sym_{i-2})=\rfact\theta {i-2}$.\label{fig:desc}}
\end{figure}

\begin{corollary}\label{thm:nb_of_descents}
Under the Ewens-like distribution on $\sym_n$ for records with parameter $\theta$, the expected value of the number of descents is: 
$\EE_{n}[\desc]= \frac{n(n-1)}{2(\theta  + n-1)}$.
\end{corollary}
In the second row of Table~\ref{table:asymptotics}, remark that the only way of obtaining a sublinear number of descents is to take very large values for $\theta$.

\medskip

Finally, we study the expected value of $\sigma(1)$. 
We are interested in this statistic to show a proof that differs from the ones for
the numbers of records and descents: the expected value of the first element of a permutation is not obtained using Lemma~\ref{lem:qot}. 

\begin{lemma}\label{lem:first_value>k}
Under the Ewens-like distribution on $\sym_n$ for records with parameter $\theta$, 
for any $k \in [0,n-1]$, 
the probability that a permutation starts with a value larger than $k$ is: 
$\PP_{n}(\sigma(1)>k) = \frac{(n-1)!\rfact\theta{n-k}}{(n-k-1)!\,\rfact\theta n}$. 
\end{lemma}

\begin{proof}
Let $\F_{n,k}$ denote the set of permutations of $\sym_{n}$ such that $\sigma(1)>k$. 
Such a permutation can uniquely be obtained by choosing the preimages of the elements in $[k]$
in $\{2,\ldots,n\}$, then by mapping bijectively the remaining elements to
$[k+1,n]$. Since none of the elements in $[k]$ is a record and since
the elements of $[k+1,n]$ can be ordered in all possible ways, we get
that $\weight(\F_{n,k}) =\binom{n-1}{k}k!\,{\rfact \theta {n-k}}$. Indeeed,there are
$\binom{n-1}{k}k!$ ways to position and order the elements of $[k]$, and the total weight of the elements larger than $k$ is ${\rfact \theta {n-k}}$.
Hence, $\PP_{n}(\sigma(1)>k) =\frac{\weight(\F_{n,k})}{\weight(\sym_n)} =\frac{\binom{n-1}{k}k! \rfact \theta {n-k}}{\rfact{\theta}{n}} = \frac{(n-1)!\rfact\theta{n-k}}{(n-k-1)!\,\rfact\theta n}$.
\end{proof}

\begin{theorem}\label{cor:first_value=k}
Under the Ewens-like distribution on $\sym_n$ for records with parameter $\theta$, 
for any $k \in [n]$, 
the probability that a permutation starts with $k$ is: 
$\PP_{n}(\sigma(1)=k) = \frac{(n-1)!\,\rfact{\theta}{n-k}\theta}{(n-k)!\rfact{\theta}{n}}$.
\end{theorem}

\begin{corollary}\label{thm:first_value}
Under the Ewens-like distribution on $\sym_n$ for records with parameter $\theta$, the expected value of the first element of a permutation is: 
$\EE_{n}[\sigma(1)] = \frac{\theta + n}{\theta + 1}$.
\end{corollary}

\begin{remark}
Our proof of Corollary~\ref{thm:first_value} relies on calculus, but gives a very simple expression for $\EE_{n}[\sigma(1)]$. 
We could therefore hope for a more combinatorial proof of Corollary~\ref{thm:first_value}, but we were not able to find it.
\end{remark}

\subsection{Number of inversions and expected running time of \InsertSort}\label{sec:insertion}

Recall that an \emph{inversion} in a permutation $\sigma\in\sym_{n}$ is a pair $(i,j)\in[n]\times[n]$ such that $i<j$ and $\sigma(i) > \sigma(j)$. 
In the word representation of permutations, this corresponds to a pair of elements in which the largest is to the left of the smallest. 
This equivalent definition of inversions naturally generalizes to sequences of distinct integers. 
For any $\sigma\in\sym_{n}$, we denote by $\inv(\sigma)$ the number of inversions of $\sigma$, 
and by $\inv_j(\sigma)$ the number inversions of the form $(i,j)$ in $\sigma$, for any $j\in [n]$. 
More formally, $\inv_{j}(\sigma) =\big| \{i\in[j-1]:\ (i,j)\text{ is an inversion of }\sigma\}\big|$.

\begin{theorem}\label{lem:inversion_at_j}
Under the Ewens-like distribution on $\sym_n$ for records with parameter $\theta$, 
for any $j \in [n]$ and $k\in[0,j-1]$, 
the probability that there are $k$ inversions of the form $(i,j)$ is: 
$\PP_{n}\big(\inv_j(\sigma) =k\big) =  \frac{1}{\theta+j-1}$ if $k\neq 0$ 
and $\PP_{n}\big(\inv_j(\sigma) =k\big) =\frac{\theta}{\theta+j-1}$ if $k=0$. 
\end{theorem}

\begin{corollary}\label{thm:nb_of_inversions}
Under the Ewens-like distribution on $\sym_n$ for records with parameter $\theta$, the expected value of the number of inversions is: 
$\EE_{n}[\inv] = \frac{n(n + 1-2\theta)}4 + \frac{\theta(\theta-1)}{2}\Delta(\theta,n)$.
\end{corollary}

Recall that the \InsertSort algorithm works as follows: at each step
$i\in\{2,\ldots,n\}$, the first $i-1$ elements are already sorted, and the
$i$-th element is then inserted at its correct place, by swapping the needed elements. 

It is well known that the number of swaps performed by \InsertSort when applied to $\sigma$ is equal to the number of inversions $\inv(\sigma)$ of $\sigma$. 
Moreover, the number of comparisons $C(\sigma)$ performed by the algorithm satisfies $\inv(\sigma)\leq C(\sigma)\leq \inv(\sigma)  +  n -1$ 
(see~\cite{CoLeRi01} for more information on \InsertSort).

As a direct consequence of Corollary~\ref{thm:nb_of_inversions} and the asymptotic estimates of the fourth row of Table~\ref{table:asymptotics}, 
we get the expected running time of \InsertSort:
\begin{corollary}\label{cor:insert sort}
Under the Ewens-like distribution for records with parameter $\theta=\O(n)$, the expected running time of \InsertSort is $\Theta(n^{2})$, like under the uniform distribution. 
If $\theta=n^{\delta}$ with $1<\delta<2$, it is $\Theta(n^{3-\delta})$. If $\theta=\Omega(n^{2})$, it is $\Theta(n)$.
\end{corollary}

\section{Expected Number of Mispredictions for the Min/Max Search}\label{sec:mispredictions}

\subsection{Presentation}\label{sec:presentation}
In this section, we turn our attention to a simple and classical problem: computing both the minimum and the maximum of an array of size~$n$. The straightforward approach (called {\em naive} in the sequel) is to compare all the elements of the array to the current minimum and to the current maximum, updating them when it is relevant. This is done\footnote{Note that, for consistency, our arrays start at index~1, as stated at the beginning of this paper.} in Algorithm~\ref{algo:naiveMinMax} and uses exactly $2n-2$ comparisons.
A classical optimization is to look at the elements in pairs, and to compare the smallest to the current minimum and the largest to the current maximum (see Algorithm~\ref{algo:32MinMax}). This uses only $3n/2$ comparisons, which is optimal.
However, as reported in~\cite{AuNiPi16}, with an implementation in~C of these two algorithms, the naive algorithm proves to be the fastest on uniform permutations as input. The explanation for this is a trade-off between the number of comparisons involved and an other inherent but less obvious factor that influences the running time of these algorithms: the behavior of the branch predictor.
\SetKw{KwBy}{by} 

\noindent
\begin{minipage}{.41\textwidth}
\begin{small}
\begin{algorithm}[H]
\DontPrintSemicolon
  $min \gets T[1]$\;
  $max \gets T[1]$\;
  \smallskip
  \For{ i $\gets$ 2 \KwTo $n$ }{    \nllabel{line:naiveFor1}
    \If{ $T[i] < min$ }{              \nllabel{line:naiveIf1}
      $min \gets T[i]$
    }
    \If{ $T[i] > max$ }{              \nllabel{line:naiveIf2}
      $max \gets T[i]$
    }
  }                                   \nllabel{line:naiveFor2}

  \Return $min,max$
\caption{ \textsc{naiveMinMax}$(T,n)$ \label{algo:naiveMinMax}}
\end{algorithm}
\end{small}

\medskip
\end{minipage}  
\hfill
\begin{minipage}{.54\textwidth}
\vspace{0pt}  
\begin{algorithm}[H]
\begin{small}
\DontPrintSemicolon
  $min,max \gets T[n],T[n]$\;
  \For{ i $\gets$ 2 \KwTo $n$ \KwBy $2$ }{
    \If{ $T[i-1] < T[i]$ }{                       \nllabel{line:3demiIfPrincipal}
	  $pMin, pMax \gets T[i-1], T[i]$      
    }
    \lElse{
      $pMin, pMax \gets T[i], T[i-1]$
    }
    \smallskip
    \lIf{ $pMin < min$ }{$min \gets pMin$ }   \nllabel{line:3demiIfMin}
    \smallskip
	\lIf{ $pMax > max$ }{$max \gets pMax$ }   \nllabel{line:3demiIfMax}  
  }
  \Return $min,max$
\caption{ \textsc{3/2-MinMax}$(T,n)$ \label{algo:32MinMax}}
\end{small}
\end{algorithm}
\medskip
\end{minipage}

In a nutshell, when running on a modern processor, the instructions that constitute a program are not executed strictly sequentially but instead, they usually overlap one another since most of the instructions can start before the previous one is finished. This mechanism is commonly described as a {\em pipeline} (see~\cite{HePa11} for a comprehensive introduction on this subject). However, not all instructions are well-suited for a pipelined architecture: this is specifically the case for branching instructions such as an \emph{if} statement. When arriving at a branch, the execution of the next instruction should be delayed until the outcome of the test is known, which stalls the pipeline. To avoid this, the processor tries to predict the result of the test, in order to decide which instruction will enter the pipeline next.
If the prediction is right, the execution goes on normally, but in case of a {\em misprediction}, the pipeline needs to be flushed, which can significantly slow down the execution of a program.

There is a large variety of branch predictors, but nowadays, most processors use {\em dynamic} branch prediction: 
they remember partial information on the results of the previous tests at a given \emph{if} statement, 
and their prediction for the current test is based on those previous results.
These predictors can be quite intricate, but in the sequel, we will only consider local {\em 1-bit predictors} which are state buffers associated~to each \emph{if} statement: they store the last outcome of the test and guess that the next outcome will be the same. 

Let us come back to the problem of simultaneously finding the minimum and the maximum in an array. 
We can easily see that, for Algorithm~\ref{algo:naiveMinMax}, the behavior of a 1-bit predictor when updating the maximum (resp. minimum) 
is directly linked to the succession of records (resp. min-records\footnote{
A \emph{min-record} (a.k.a. left to right minimum) is an element of the array such that no smaller element appears to its left.}) 
in the array. 
As we explain later on, for Algorithm~\ref{algo:32MinMax}, this behavior depends on the ``pattern'' seen 
in four consecutive elements of the array, this ``pattern'' indicating not only which elements are records (resp. min-records), 
but also where we find descents between those elements. 
As shown in~\cite{AuNiPi16}, for uniform permutations, Algorithm~\ref{algo:naiveMinMax} outerperforms Algorithm~\ref{algo:32MinMax}, 
because the latter makes more mispredictions than the former, compensating for the fewer comparisons made by Algorithm~\ref{algo:32MinMax}. 
This corresponds to our Ewens-like distribution for $\theta=1$. 
But when $\theta$ varies, the way records are distributed also changes, influencing the performances of both Algorithms~\ref{algo:naiveMinMax} and~\ref{algo:32MinMax}. 
Specifically, when $\theta=\lambda n$, we have a linear number of records (as opposed to a logarithmic number when $\theta=1$). 
The next subsections provide a detailed analysis of the number of mispredictions in Algorithms~\ref{algo:naiveMinMax} and~\ref{algo:32MinMax}, 
under the Ewens-like distribution for records, with a particular emphasis on $\theta=\lambda n$
(which exhibits a very different behavior w.r.t. the uniform distribution -- see Fig.~\ref{fig:compare}).

\subsection{Expected Number of Mispredictions in NaiveMinMax}\label{sec:naive}

\begin{theorem} \label{thm:misprediction_naive}
Under the Ewens-like distribution on $\sym_n$ for records with parameter $\theta$, 
the expected numbers of mispredictions at lines~\ref{line:naiveIf1} and~\ref{line:naiveIf2} of Algorithm~\ref{algo:naiveMinMax} satisfy respectively
$\EE_{n}[\mu_{4}] \leq \frac2\theta \EE_{n}[\rec] $ and 
$\EE_{n}[\mu_{6}] =2\theta\Delta(\theta,n-1)-\frac{(2\theta + 1)(n-1)}{\theta + n-1}$.
\end{theorem}
Consequently, with our previous results on $\EE_{n}[\rec] $, 
the expected number of mispredictions at line~\ref{line:naiveIf1} is $\O(\log n)$ when $\theta=\Omega(1)$ (\emph{i.e.}, when $\theta = \theta(n)$ is constant or larger). 
Moreover, using the asymptotic estimates of the digamma function, 
the asymptotics of the expected number of mispredictions at line~\ref{line:naiveIf2} is such that (again, for $\lambda>0$, $0<\epsilon<1$ and $\delta>1$): 
\begin{center}
\begin{tabular}{c|c|c|c|c}
& fixed $\theta>0$ & $\theta := n^{\epsilon}$  & $\theta :=  \lambda n$ & $\theta := n^{\delta}$ \\
\hline 
$\EE_{n}[\mu_{6}]$ & $\sim 2\theta\cdot \log n $ & $\sim 2(1-\epsilon)\cdot n^{\epsilon}\log n$ & $\sim 2\lambda(\log(1 + 1/\lambda)-1/(\lambda + 1))\cdot n$ & $ o(n)$
\end{tabular}
\end{center}
In particular, asymptotically, the expected total number of mispredictions of Algorithm~\ref{algo:naiveMinMax} is given by $\EE_{n}[\mu_{6}]$ 
(up to a constant factor when $\theta$ is constant). 

\subsection{Expected Number of Mispredictions in \texorpdfstring{$\frac32$}-MinMax}\label{sec:optimal}

Mispredictions in Algorithm~\ref{algo:32MinMax} can arise in any of the three \emph{if} statements. 
We first compute the expected number of mispredictions at each of them independently. 
We start with the \emph{if} statement of line~\ref{line:3demiIfPrincipal}, which compares $T[i-1]$ and $T[i]$. For our 1-bit model, there is a misprediction whenever
there is a descent at $i-2$ and an ascent at $i$, or an ascent at $i$ and a descent at $i-2$. A tedious study of all possible cases gives:

\begin{theorem}\label{thm:32-first-if}
Under the Ewens-like distribution on $\sym_n$ for records with parameter $\theta$, 
the expected number of mispredictions at line~\ref{line:3demiIfPrincipal} of Algorithm~\ref{algo:32MinMax} satisfies
\begin{align*}\textstyle
\EE_{n}[\nu_{3}] = &  \frac{n-2}{4} + \frac{\theta(\theta-1)^2}{4} + \frac{\theta^2(\theta-1)^2}{12} \left( \frac{1}{\theta+n-1} -\frac{3}{\theta+n-2} -\frac{1}{\theta+1}\right)\\
& +  \frac{\theta^2(\theta-1)^2}{6} \left(\Delta\left(\frac{\theta+1}2,\frac{n-2}{2}\right)-\Delta\left(\frac{\theta}2,\frac{n-2}{2}\right)\right).
\end{align*}
As a consequence,
if $\theta=\lambda n$, then $\EE_{n}[\nu_{3}] \sim \frac{6\lambda^{2} + 8\lambda + 3}{12(\lambda + 1)^{3}}\,n$. 
\end{theorem}

\begin{theorem}\label{thm:32-second-if}
Under the Ewens-like distribution on $\sym_n$ for records with parameter $\theta$, 
the expected number of mispredictions at line~\ref{line:3demiIfMin} of Algorithm~\ref{algo:32MinMax} satisfies
$\EE_{n}[\nu_{7}]\leq \frac2{\theta} \EE_{n}[\rec] $. As a consequence, if $\theta=\lambda n$, then $\EE_{n}[\nu_{7}] = \O(1)$. 
\end{theorem}

We now consider the third \emph{if} statement of Algorithm~\ref{algo:32MinMax}. 
If there is a record (resp. no record) at position $i-3$ or $i-2$, then there is a misprediction when there is no record (resp. a record) at position $i-1$ or $i$. Studying all the possible configurations at these four positions gives the following result. 
\begin{theorem}\label{thm:32-third-if}
Under the Ewens-like distribution on $\sym_n$ for records with parameter $\theta$, 
the expected number of mispredictions at line~\ref{line:3demiIfMax} of Algorithm~\ref{algo:32MinMax} satisfies
\begin{align*}\textstyle
\EE_{n}[\nu_{8}] = & \frac{(n-2)((2\theta^{3}+\theta^{2}-9\theta-3)n+2\theta^{4}-5\theta^{2}+9\theta+3)}{3(\theta+n-1)(\theta+n-2)}\\
&+ \frac{\theta(2\theta^{3}+\theta+3)}{3}\ \Delta\left(\frac{\theta+1}{2},\frac{n-2}{2}\right)
- \frac{\theta(2\theta^{3}+\theta-3)}{3}\ \Delta\left(\frac{\theta}{2},\frac{n-2}{2}\right).
\end{align*}
As a consequence,
if $\theta=\lambda n$, then $\EE_{n}[\nu_{7}] \sim \big(2\lambda\log\left(1+\frac1\lambda\right)
-\frac{\lambda(6\lambda^{2}+15\lambda+10)}{3(\lambda+1)^{3}} \big)n$.
\end{theorem}
It follows from Theorems~\ref{thm:32-first-if},~\ref{thm:32-second-if} and~\ref{thm:32-third-if} that:
\begin{corollary}
Under the Ewens-like distribution on $\sym_n$ for records with parameter $\theta=\lambda n$, 
the total number of mispredictions of Algorithm~\ref{algo:32MinMax} is 
\[
\EE_{n}[\nu] \sim\left(2\lambda\log\left(1+\frac1\lambda\right)-\frac{24\lambda^{3}+54\lambda^{2}+32\lambda-3}{12(\lambda+1)^{3}}\right)\,n. 
\]
\end{corollary}

Fig.~\ref{fig:compare} shows that, unlike in the uniform case ($\theta=1$), Algorithm~\ref{algo:32MinMax} is more efficient than Algorithm~\ref{algo:naiveMinMax} 
under the Ewens-like distribution for records with $\theta:=\lambda n$, as soon as $\lambda$ is large enough.
\begin{figure}[ht]
\begin{minipage}{.35\textwidth}
\vspace*{0.5em}
\includegraphics[scale=.98]{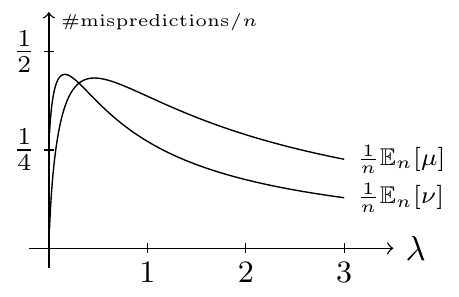}
\end{minipage}
\begin{minipage}{.63\textwidth}
\caption{The expected number of mispredictions produced by the naive algorithm ($\mu$) and for $\frac32$-minmax ($\nu$), when $\theta:=\lambda n$. We have
$\EE_{n}[\mu] \sim \EE_{n}[\nu]$ for $\lambda_{0}=\frac{\sqrt{34}-4}6\approx 0.305$, 
and there are fewer mispredictions on average with $\frac32$-minmax as soon as $\lambda > \lambda_{0}$.
However, since $\frac32$-minmax performs $\frac{n}2$ fewer comparisons than the naive algorithm, it becomes more efficient before $\lambda_{0}$. 
For instance, if a misprediction is worth $4$ comparisons, $\frac32$-minmax is the most efficient as soon as $\lambda>0.110$.\label{fig:compare}}
\end{minipage}
\end{figure}

\paragraph{Acknowledgments} Many thanks to Valentin Féray for providing insight and references on the classical Ewens distribution and to the referees whose comments helped us clarify the presentation of our work.

\bibliographystyle{plain}
\begin{small}
\bibliography{aofa2016_aubonipi}
\end{small}

\end{document}